\documentclass[10pt, a4paper]{article}
\usepackage{graphicx}
\usepackage{amsmath}
\usepackage{algorithm2e}
\usepackage{url}
\usepackage{amssymb,amsmath}
\usepackage{amsthm}

\newtheorem{theorem}{Theorem}
\newtheorem{proposition}{Proposition}
\long\def\omitit#1{}
\begin{document}

\title{A Multi-Agent Prediction Market based on Partially Observable Stochastic Game }

\author{Janyl Jumadinova (1), Prithviraj Dasgupta (1) \\
  ((1) University of Nebraska at Omaha, \\Computer Science Department, Omaha NE)
}
\maketitle

We present a novel, game theoretic representation of
a multi-agent prediction market using a partially observable
stochastic game with information (POSGI).  We then describe a
correlated equilibrium (CE)-based solution strategy for this
game which enables each agent to dynamically calculate the
prices at which it should trade a security in the prediction market.
We have extended our results to risk averse traders and shown that
a Pareto optimal correlated equilibrium strategy can be used to
incentively truthful revelations from risk averse agents. Simulation
results comparing our CE strategy with five other strategies commonly
used in similar markets, with both risk neutral and risk averse agents,
show that the CE strategy improves price predictions and provides higher
utilities to the agents as compared to other existing strategies.


\section{Introduction}

Forecasting the outcome of events that will happen in the future
is a frequently indulged and important task for humans.
It is encountered in various domains such as forecasting
the outcome of geo-political events, betting on the outcome
of sports events, forecasting the prices of financial
instruments such as stocks, and casual predictions
of entertainment events. Despite the ubiquity of such
forecasts, predicting the outcome of future events is
a challenging task for humans or even computers -
it requires extremely complex calculations involving
a reasonable amount of domain knowledge, significant amounts of
information processing and accurate reasoning. Recently,
a market-based paradigm called {\em prediction markets} has shown
ample success to solve this problem by using the aggregated
`wisdom of the crowds' to predict the outcome of future events.
This is evidenced from the
successful predictions of actual events done by prediction markets run
by the Iowa Electronic Marketplace(IDEM), Tradesports,
Hollywood Stock Exchange, the Gates-Hillman market \cite{Othman10},
and by companies such as Hewlett Packard \cite{Plott02},
Google \cite{Cowgill09} and Yahoo's Yootles \cite{YootlesURL}.
A prediction market for a real-life
event (e.g., ``Will Obama win the 2008 Democratic Presidential
nomination?'') is run for several
days before the event happens. The event has a binary outcome (yes/no or $1/0$).
On each day, humans, called \textit{traders}, that are interested in the outcome of the event
express their belief on the possible outcome of the event
using the available information related to the event.
The information available to the different traders is asymmetric,
meaning that different traders can possess bits and pieces of
the whole information related to the event. The belief values
of the traders about the outcome of the event
are expressed as probabilities. A special entity
called the {\em market maker} aggregates the probabilities from all
the traders into a single probability value that represents
the possible outcome of the event. The main idea behind
the prediction market paradigm is that the collective, aggregated
opinions of humans on a future event represents the probability of occurrence
of the event more accurately than corresponding surveys and opinion polls.
Despite their overwhelming success, many aspects of
prediction markets such as a formal representation
of the market model, the strategic behavior of the market's
participants and the impact of information from external
sources on their decision making have not been analyzed
extensively for a better understanding. We attempt to
address this deficit in this paper by developing
a game theoretic representation of the traders' interaction
and determining their strategic behavior using the equilibrium
outcome of the game.
We have developed a
correlated equilibrium (CE)-based solution strategy for a partially
observable stochastic game representation of the prediction market.
We have also empirically compared our CE strategy with five other strategies commonly
used in similar markets, with both risk neutral and risk averse agents,
and showed that the CE strategy improves price predictions and provides higher
utilities to the agents as compared to other existing strategies.

\section{Related Work}
Prediction markets were started in 1988 at the Iowa Electronic Marketplace \cite{IEMURL} to
investigate whether betting on the outcome of gee-political events (e.g. outcome of presidential
elections, possible outcome of international political or military crises, etc.) using real
money could elicit more accurate information about the event's outcome than regular
polls. Following the success of prediction markets in eliciting information about events'
outcomes, several other prediction markets have been started that trade on events using either real or virtual money.
Prediction markets have been used in
various scenarios such as predicting the outcome of gee-political
events such as U.S. presidential elections, determining the outcome
of sporting events, predicting the box office performance of
Hollywood movies, etc.
Companies such as Google, Microsoft, Yahoo and Best Buy have all
used prediction markets internally to tap the collective
intelligence of people within their organizations. The seminal work
on prediction market analysis
\cite{Gjerstad06, Wolfers04, Wolfers06} has shown that the mean belief
values of individual traders about the outcome of a future event
corresponds to the event's market price. Since then researchers have
studied prediction markets from different perspectives. Some
researchers have studied traders' behavior by modeling their
interactions within a game theoretic framework.
For example, in \cite{Chen06, Feigenbaum05} the authors have used a Shapley-Shubik
game that involves behavioral assumptions on the agents
such as myopic behavior and truthful revelations and
theoretically analyzed the aggregation function and convergence
in prediction markets.
Chen et. al. \cite{YChen06}
 characterized the uncertainty
of market participants' private information by incorporating aggregate uncertainty in
their market model. However, both these models consider traders that are risk-neutral,
myopic and truthful.
Subsequently,
in \cite{Chen09} the authors have relaxed some of these assumptions within a Bayesian game setting
and investigated the conditions under
which the players reveal their beliefs truthfully.
Dimitrov and Sami \cite{Dimitrov08} have also studied the effect of non-myopic revelations by trading
agents in a prediction market and concluded that myopic strategies are almost never
optimal in the market with the non-myopic traders and there is a need for discounting.
Other researchers have focused on designing rules that a
market maker can use to combine the opinions (beliefs) from different
traders. Hanson \cite{Hanson07} developed a market scoring rule that is used to reward
traders for making and improving a prediction about the outcome of an event. He
further showed how any proper scoring rule can serve as an automated market maker.
 Das \cite{Das08} studied the effect of specialized agents called market-makers which behave
as intermediaries to absorb price shocks in the market. Das empirically studied different
market-making strategies and concludes that a heuristic strategy that adds a random
value to zero-profit market-makers improves the profits in the markets.

The main contribution of our paper, while building
on these previous directions, is to use a
partially observable stochastic game (POSG)
\cite{Hansen04} that can be used by each agent to
reason about its actions. Within this POSG model,
we calculate the correlated equilibrium strategy
for each agent using the aggregated price from
the market maker as a recommendation signal.
We have also considered the risk preferences
of trading agents in prediction markets and
shown that a Pareto optimal correlated equilibrium solution
can incentively truthful revelation
from risk averse agents.
We have compared the POSG/correlated equilibrium based pricing strategy
with five different pricing strategies used in
similar markets with pricing data obtained from real prediction
prediction market events. Our results show that the agents
using the correlated equilibrium strategy profile
are able to predict prices that are closer to the
actual prices that occurred in real markets and these
traders also obtain $35-127\%$ higher profits.

\section{Preliminaries}

{\bf Prediction Market.} Our prediction market consists of $N$ traders, with each
trader being represented by a software trading agent that performs actions
on behalf of the human trader. The market also has
a set of future events whose outcome has not yet been determined.
The outcome of each event is considered
as a binary variable with the outcome being $1$ if the
event happens and the outcome being $0$ if it does not.
Each outcome has a security associated with it.
A {\em security} is a contract that yields payments based on the outcome
of an uncertain future event. Securities can be purchased or sold by trading
agents at any time during the lifetime of the security's event.
A single event can have multiple securities
associated with it. Trading agents can purchase or sell one or more
of the securities for each event at a time.
A security expires when the event $e$
associated with it happens at the end of the event's
duration. At this point the
outcome of the event has just been determined
and all trading agents are notified of the event's outcome.
The trading agents that had
purchased the security during the lifetime of the event
then get paid $\$1$ if the event happens with an outcome of $1$,
or, they do not get paid anything and
lose the money they had spent on buying the security
if the event happens with an outcome of $0$.
On each day, a trading agent makes
a decision of whether to buy some securities related
to ongoing events in the market, or whether to sell or
hold some securities it has already purchased.

\begin{figure}[h]
\begin{center}
\includegraphics[width=3.5in]{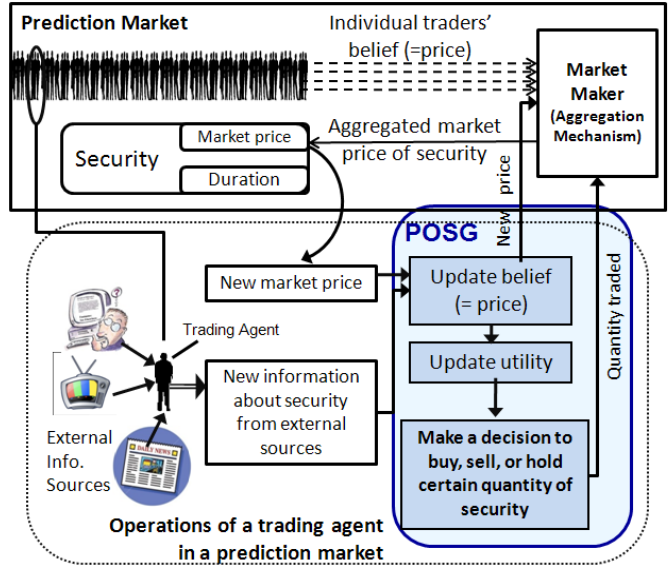}
\end{center}
\caption{Essential operations performed by a trading
agent in a prediction market and the portion of
the prediction market where the POSG representation
is applied.}
\label{fig_predmarket}
\end{figure}

{\bf Market Maker.} \label{sec_marketmaker}
Trading in a prediction market is similar to the
continuous double auction (CDA) protocol. However,
using the CDA directly as the trading protocol
for a prediction market leads to some problems
such as the \textit{thin market problem}
(traders not finding a trading partner
immediately) and traders potentially losing profit
because of revealing their willingness to trade
beforehand to other traders. To address these problems,
prediction markets
use a trading intermediary called the \textit{market maker}.
The market maker aggregates the buying and selling prices of
a security reported in the trades from the trading agents (they can
be identical) and `posts' these prices in the market.
Trading agents interact with the market maker to buy
and sell securities, so that they do not have to wait
for another trading agent to arrive before they can trade.
A market maker uses a market scoring rule (MSR) to calculate
the aggregated price of a security.
Recently, there has been considerable interest in analyzing the
MSRs \cite{Chen07,Hanson07,OthmanSandholmEC10}
and an MSR called the logarithmic
MSR (LMSR) has been shown to guarantee truthful revelation of beliefs
by the trading agents \cite{Hanson07}. LMSR allows a security's price, and
payoffs to agents buying/selling the security, to be expressed
in terms of its purchased or outstanding quantity. Therefore,
the `state' of the securities in the market can be captured
only using their purchased quantities.
The market-maker in our prediction market uses LMSR to update the market price
and to calculate the payoffs to trading agents. We briefly describe the
basic mechanism under LMSR.
Let $\Xi$ be the set of securities
in a prediction market and $\overline{q} = (q_1, q_2...q_{|\Xi|})$ denote
the vector specifying the number of units of each security held by
the different trading agents at time $t$. The LMSR first calculates
a cost function to reflect the total money wagered
in the prediction market by the trading agents as:
$C(\overline{q}) = b \cdot ln(\sum_{j=1}^{\mid \Xi \mid} e^{q_j/b})$.
It then calculates the aggregated market price $\pi$ for the
security $\xi \in \Xi$ as:
$\pi_{\xi} = e^{(q_\xi/b)}/ \sum_{j=1}^{\mid \Xi \mid} (e^{(q_j/b)})$ \cite{Chen07}
{\footnote {Parameter \textit{b} (determined by the market maker)
controls the monetary risk of the market maker as well as the quantity
of shares that trading agents can trade
at or near the current price without causing massive price swings.
Larger values for $b$ allows trading agents
to trade more frequently but also increases the market maker's
chances to lose money.}}.
Trading agents inform the quantity of a security they
wish to buy or sell to the market maker.
If a trading agent purchases $\delta$ units
of a security, the market maker determines the payment
the agent has to make as $C(\overline{q}+\delta) - C(\overline{q})$.
Correspondingly, if the agent sells $\delta$ quantity
of the security, it receives a payoff of
$C(\overline{q}) - C(\overline{q}-\delta)$
from the market maker.
Figure \ref{fig_predmarket} shows the operations
performed by a trader (trading agent) in a prediction market and
the portions of the prediction market that are affected by the POSG
representation described in the next section.

\section{Partially Observable Stochastic Games for Trading
Agent Interaction}
\label{POSGI}
For simplicity of explanation, we consider
a prediction market where a single security is being traded over
a certain duration. This duration is divided into trading periods,
with each trading period corresponding to a day in a real prediction market.
The `state' of the market is expressed as the quantity of the purchased
units of the security in the market. At the end of
each trading period, each trading agent receives information about the state
of the market from the market maker. With this prior
information, the task of a trading agent is to determine a suitable
quantity to trade for the next trading period, so that its utility
is maximized. In this scenario, the environment of the agent
is partially observable because other agents' actions and payoffs are
not known directly, but available through their aggregated beliefs.
Agents interact with each other in stages (trading periods), and in each stage
the state of the market is determined stochastically based on
the actions of the agents and the previous state. This scenario
directly corresponds to the setting of a partially observable
stochastic game \cite{Fudenberg,Hansen04}.
A POSG model offers several attractive features such as
{\em structured behavior} by the agents by using best response
strategies, {\em stability of the outcome} based on equilibrium
concepts, {\em lookahead capability} of the agent to plan their
actions based on future expected outcomes,
{\em ability to represent the temporal characteristics} of the interactions
between the agents, and, enabling all computations locally on the
agents so that the system is {\em robust and scalable}.

Previous research has shown that
information related parameters in a prediction market
such as information availability, information reliability,
information penetration, etc., have a considerable effect
on the belief (price) estimation
by trading agents. Based on these findings, we
posit that a component to model the impact of information related
to an event should be added to the POSG framework.
With this feature in mind, we propose an interaction model
called a partially observable stochastic game with information (POSGI)
for capturing the strategic decision making by trading agents.
A POSGI is defined as: \\$\Gamma = (N, S, (A_i)_{i \in N}, (R_i)_{i \in
  N}, T, (O_i)_{i \in N}, \Omega, ({\cal I}_i)_{i \in N})$, where
$N$ is a finite set of agents, $S$ is a finite, non-empty set of
states - each state corresponding to certain quantity of the security
being held (purchased) by the trading agents.
$A_i$ is a finite non-empty action space of agent $i$ s.t.
$\overline{a_k} = (a_{1,k},...,a_{|N|,k})$ is the joint action of the agents and
$a_{i,k}$ is the action that agent $i$ takes in state $k$.
In terms of the prediction market, a trading agent's action
corresponds
to certain quantity of security it buys or sells, while the joint
action corresponds to changing the purchased quantity for a security
and taking the market to a new state.
$R_{i,k}$ is the reward or payoff for agent $i$ in state $k$
which is calculated using the LMSR market maker.
$T: T(s,\overline{a},s') = P(s'|s,\overline{a})$ is the transition
probability of moving from state $s$ to state $s'$
after joint action $\overline{a}$ has been performed by the agents.
$O_i$ is a finite non-empty set of observations for agent $i$
that consists of the market price and the information signal,
and $o_{i,k} \in O_i$ is the observation agent $i$ receives in state $k$.
$\Omega: \Omega(s_k,I_{i,k},o_{i,k}) = P(o_{i,k}|s_k,I_{i,k})$ is the
observation probability for agent $i$ of receiving observation $o_{i,k}$
in state $s_k$ when the information signal is $I_{i,k}$.
Finally, ${\cal I}_i$ is the information set received by agent $i$
for an event
${\cal I}_i = \bigcup_k I_{i,k}$ where $I_{i,k} \in \{-1, 0, +1\}$ is the information
received by agent $i$ in state $k$. The complete information arriving
to the market ${\cal I} = \bigcup_{i \in N} {\cal I}_i$ is
temporally distributed over the duration of the event, and,
following the information arrival patterns observed
in stock markets  \cite{Maheu04}, we assume that
new information arrives following a Poisson distribution.
Information that improves the probability of the positive outcome
of the event is considered positive($I_{i,k}=+1$) and vice-versa,
while information that does not affect the probability is
considered to have no effect ($I_{i,k}=0$).
For example, for a security related to the event
``Obama wins 2008 presidential elections'',
information about Oprah Winfrey endorsing Obama
would be considered high impact positive information and information
about Obama losing the New Hampshire Primary would
be considered negative information.

\begin{figure}[h]
\begin{center}
\includegraphics[width=3.5in,angle=-90]{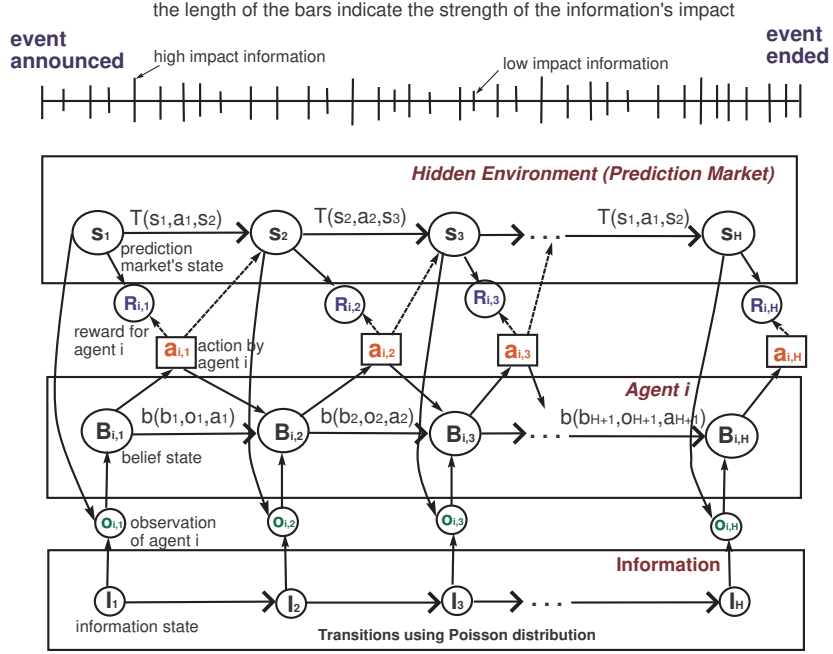}
\end{center}
\caption{An agent interactions with the hidden environment (prediction
  market) and an external information source.}
\label{automata}
\end{figure}

Based on the POSGI formulation of the prediction market,
the interaction of an agent with the environment (prediction market) and the information
source can be represented by the transition diagram shown
in Figure \ref{automata}{\footnote{We have only
shown one agent $i$ to keep the diagram legible, but
the same representation is valid for every agent in the prediction
market. The dotted lines represent that the reward and environment
state is determined by the joint action of all agents.}}.
The environment (prediction market) goes through a set of states
$\tilde{S} = \{s_1,...,s_H\}: \tilde{S} \in S$, where $H$ is the
duration of the event in the prediction market and $s_h$ represents the
state of the market during trading period $h$. This state of the market is not visible
to any agent. Instead, each agent $i$ has its own internal
belief state $B_{i,h}$ corresponding to its belief about the actual state $s_h$.
$B_{i,h}$ gives a probability distribution over the set of states
$S$, where $B_{i,h} = (b_{1,h},...,b_{|S|,h})$.
Consider trading period $h-1$ when the agents perform the joint action
$\overline{a}_{h-1}$. Because of this joint action of the agents
the environment stochastically changes to a new state
$s_{h}$, defined by the state transition function $T(s_{h-1}, a_{h-1}, s_h)$.
There is also an external information state,
$I_h$, that transitions to the next state
$I_{h+1}$ given by the Poisson distribution \cite{Maheu04} from
which the information signal is sampled.
The agent $i$ doesn't directly see the
environment state, but instead receives an
observation $o_{i,S_h} = (\pi_{s_h}, {\cal I}_{i,s_h})$, that
includes the market price $\pi_{s_h}$ corresponding
to the state $s_h$ as informed by the market maker,
and the information signal ${\cal I}_{i,s_h}$. The agent $i$ then
uses a \textit{belief update function} to
update its beliefs. Finally, agent $i$ selects an action
using an \textit {action selection strategy} and receives
a reward $R_{i,s_h}$.

\subsection{Trading agent belief update function}
\label{beliefupdate}
Recall from Section \ref{POSGI} that a belief state of a trading agent is a probability vector that
gives a distribution over the set of states $S$ in the prediction market, i.e.
$B_{i,h} = (b_{1,h},...,b_{|S|,h})$.
A trading agent uses its belief update function
$b: \Re^{\mid S \mid} \times A_i \times O_i \rightarrow \Re^{\mid S \mid}$
to update its belief state based on its past action $a_{i,h-1}$,
past belief state $B_{i,h-1}$ and the observation $o_{i,S_h}$.
The calculation of the belief update function
for each element of the belief state, $b_{s',h}$, $s' \in S$, is described below:
\begin{eqnarray}
b_{s',h} &=& P(s'|a_{i,h-1},o_i) = \frac{P(s',a_{i,h-1},o_i,)}{P(a_{i,h-1},o_i)} \nonumber \\
& = & \frac{P(o_i|s',a_{i,h-1}) \cdot P(s',a_{i,h-1})}{P(o_i|a_{i,h-1}) P(a_{i,h-1})}
\label{main_eqn}
\end{eqnarray}
Because $a_{i,h-1}$ is conditionally independent given $s'$ and  $o_i$ is conditionally independent given $a_{i,h-1}$, we can rewrite Equation \ref{main_eqn} as:
\begin{eqnarray}
& &b_{s',h} =  \frac{P(o_i|s') \cdot P(s',a_{i,h-1})}{P(o_i) P(a_{i,h-1})} = \nonumber \\
& & \frac{\sum_{\iota \in \cal I} P(\iota) P(o_i|s',\iota) \sum_{s \in S} P(s) P(s'|s,a_{i,h-1}) P(a_{i,h-1})} {P(o_i) P(a_{i,h-1})} \nonumber \\
& = & \frac{\sum_{\iota \in \cal I} P(\iota) P(o_i|s',\iota) \sum_{s \in S} P(s) P(s'|s,a_{i,h-1})} {P(o_i)} \nonumber \\
& = & \frac{\sum_{\iota \in \cal I} P(\iota)\Omega(s',\iota, o_i) \sum_{s \in S} T(s,a_{i,h-1},s') b_{s,h-1}}{P(o_i)}
\label{main_eqn1}
\end{eqnarray}

All the terms in the r.h.s. of the Equation \ref{main_eqn1} can be calculated by an agent: $P(\iota)$ is the
probability of receiving information signal $\iota$ and is available
from the Poisson distribution \cite{Maheu04} for the information arrival, $\Omega(s',\iota,o_i)$ is the probability of receiving observation $o_i$ in state $s'$ when the information signal is $\iota$, $T(s,a_{i,h-1},s')$ is the probability
that the state $s$ transitions to state $s'$ after agent $i$ takes action $a_{i,h-1}$, $ b_{s,h-1}$ is the past belief of agent $i$ about state $s$, $P(o_i)$ is the probability of receiving observation $o_i$, which is constant and can be viewed as a normalizing constant.

\begin{figure}[h]
\begin{center}
\includegraphics[width=3.6in]{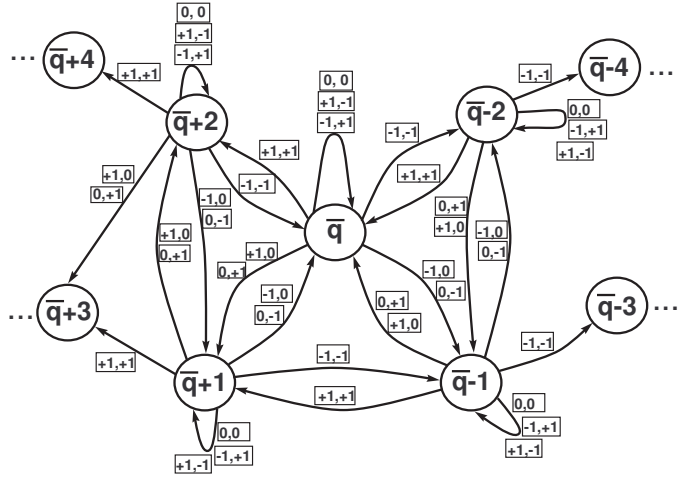}
\caption{Finite state automata of the environment
represented by the number of outstanding units of the
security, $\overline{q}$, in the prediction market. }
\label{automata_q}
\end{center}
\end{figure}

\begin{figure}[h]
\begin{center}
\includegraphics[width=3.5in]{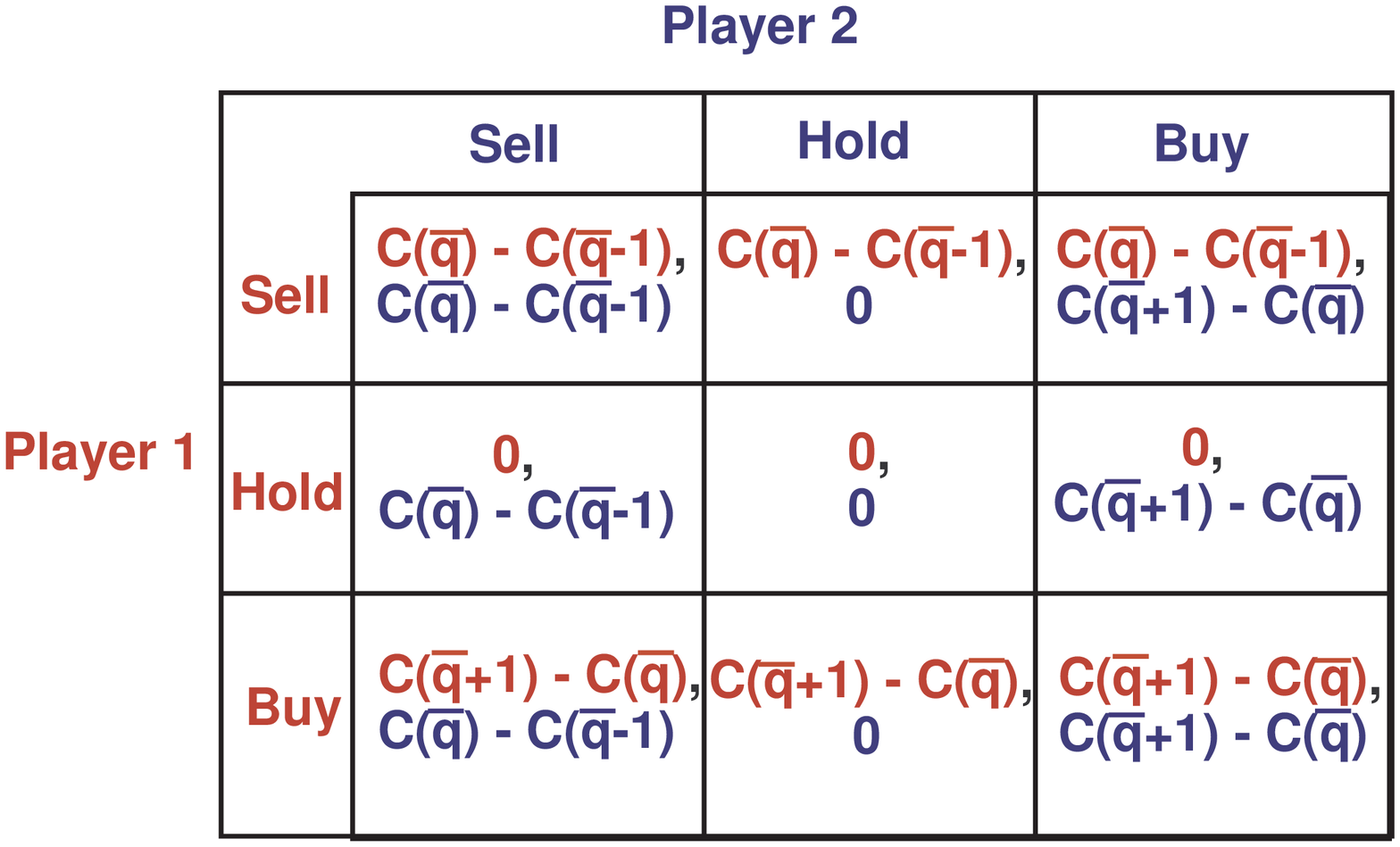}
\caption{Two player normal form game in a state for which the number
of outstanding shares is $\overline{q}$.}
\label{fig_2playergame}
\end{center}
\end{figure}

\subsection{Trading agent action selection strategy}

The objective of a trading agent in a prediction market
is to select an action at each stage so that the expected reward
that it receives is maximized. To understand this action selection process,
we consider the decision problem
facing each trading agent. Consider two agents whose
available actions during each time step are to buy (=+1) or sell(=-1) only
one unit of the security or not do anything (=0) by holding the security.
Let the market state be denoted by $\overline{q}$, the number of purchased units
of the security. Based on the set of actions available to
each agent, the state can transition to one of the following states
$\overline{q}+2$ (both agents buy), $\overline{q}+1$ (only one agent buys),
$\overline{q}$ (both agents hold, or, one agent buys while the other agent sells, resulting
in no transition), $\overline{q}+1$ (only agent sells), and
$\overline{q}+2$ (both agents sell), as shown in Figure \ref{automata_q}. We
can expand this state space further by adding more states and
transitions, but the number of states remains
finite because the set of states $S$ of the POSGI is finite. Also,
since the number of units of a security is finite, the number of
transitions from a state is guaranteed to be bounded.
We can construct a normal form game from this representation
to capture the decision problem for each agent. The payoff
matrix for this 2-player game at state $\overline{q}$ is shown in
Figure \ref{fig_2playergame}. From this payoff matrix,
we can use game-theoretic analysis
to determine equilibrium strategies for each trading agent
and find the actions that maximize its expected reward, as described below.

\subsection{Correlated Equilibrium (CE) calculation}
\label{sec_CE}
In the POSGI, the aggregated price information received by a trading
agent from the market maker can be treated as a recommendation
signal for selecting the agent's strategy. This situation lends
itself to a correlated equilibrium (CE) \cite{Aumann} where a trusted
external agent privately recommends a strategy to play to each player.
A CE is more preferred to the
Nash or Bayesian Nash equilibrium because it can lead to improved payoffs,
and it can be calculated using a linear program in time polynomial
in the number of agents and number of strategies.

Each agent $i$ has a finite set of
strategy profiles, $\Phi_i$ defined over its action space $A_i$.
The joint strategy space
is given by $\Phi = \prod_{i=1}^{|N|} \Phi_i$ and let
$\Phi_{-i} = \prod_{j \neq i} \Phi_j$.
Let $\phi \in \Phi$ denote a strategy profile
and $\phi_i$ denote player $i$'s component in $\phi$.
The utility received by each agent $i$ corresponds
to its payoff or reward and is given by the agent $i$'s
utility function $u_i( \phi)=R_i$ (for legibility
we have dropped state $k$, but the same calculation applies
at every state).
A correlated equilibrium
is a distribution $p$ on $\Phi$ such that for all
agents $i$ and all strategies $\phi_i, \phi_i'$
if all agents follow a strategy profile $\phi$ that
recommends player $i$ to choose strategy $\phi_i$,
agent $i$ has no incentive to play another
strategy $\phi_i'$ instead. This implies that
the following expression holds:
$\sum_{\phi_{-i} \in \Phi_{-i}} p(\phi) (u_{i}(\phi) - u_{i}(\phi_i',
\phi_{-i})) \geq 0$,  $\forall i \in N$, $\forall \phi_i, \phi_i' \in
\Phi_i$ and where $u_i(\phi'_i,\phi_{-i})$ is the utility that agent
$i$ gets when it changes its strategy to $\phi'_i$ while all the other agents keep their strategies fixed at
$\phi_{-i}$ and $p(\phi)$
is the probability of realizing a given strategy profile $\phi$.

\begin{theorem} A correlated equilibrium exists in our
POSGI-based prediction market representation
at each stage (trading period).
\end{theorem}
\begin{proof}
At each stage in our prediction market, we can specify the correlated
equilibrium by means of linear constraints as given below:\\
\begin{equation}
\sum_{\phi_{-i} \in \Phi_{-i}} p(\phi) (u_{i}(\phi) - u_{i}(\phi_i', \phi_{-i})) \geq 0,  \forall i \in N, \forall \phi_i, \phi_i' \in \Phi_i
\label{CE}
\end{equation}
\begin{equation}
\sum_{\phi \in \Phi} p(\phi) = 1,
\label{CE_constraint1}
\end{equation}
\begin{equation}
p(\phi) \geq 0
\label{CE_constraint2}
\end{equation}
Equation \ref{CE} states that when agent $i$ is recommended to select strategy $\phi_i$, it must get no less utility from selecting strategy $\phi_i$ as it would from selecting any other strategy $\phi_i'$.
Constraints \ref{CE_constraint1} and \ref{CE_constraint2} guarantee that $p$ is a valid probability distribution.
We can rewrite the linear program specification of the correlated equilibrium above by adding an objective function to it.
\begin{equation}
\mbox{max } \sum_{\phi \in \Phi} p(\phi), \mbox{ or min }  -\sum_{\phi \in \Phi} p(\phi)  \mbox{ s.t.}
\label{CE_alternative}
\end{equation}
\begin{equation}
\sum_{\phi \in \Phi, \phi_{-i} \in \Phi_{-i}} p(\phi) (u_{i}(\phi) - u_{i}(\phi_i', \phi_{-i})) \geq 0,
\label{CE_alternative_constraint1}
\end{equation}
\begin{equation}
\phi \geq 0
\label{CE_alternative_contraint2}
\end{equation}

Program \ref{CE_alternative_constraint1} is either trivial with a
maximum of $0$ or unbounded. Hoenselaar \cite{Hoenselaar} proved that there
is a correlated equilibrium if and only if Program \ref{CE_alternative_constraint1}
is unbounded. To prove the unboundedness we consider the dual problem of
Equation \ref{CE_alternative_constraint1} given in Equation \ref{CE_dual}.
\begin{equation}
\mbox{max } 0, \mbox{ s.t.}
\label{CE_dual0}
\end{equation}
\begin{equation}
\sum_{\phi \in \Phi, \phi_{-i} \in \Phi_{-i}} p(\overline{\phi}) [(u_{i}(\phi) - u_{i}(\phi_i', \phi_{-i})]^T \leq -1
\label{CE_dual}
\end{equation}
\begin{equation}
\overline{\phi} \geq 0
\label{CE_dual_constraint}
\end{equation}
where for every $p(\overline{\phi})$ there is $p(\phi)$ such that $p(\phi) [(u_{i}(\phi) - u_{i}(\phi_i', \phi_{-i})] p(\overline{\phi}) = 0$.

Also in \cite{Papa08}, the authors showed that the problem given in Equation
\ref{CE_dual} is always infeasible.
From operations research we know that when the dual problem
is infeasible the primal problem is feasible and unbounded.
This means that the primal problem from Equation
\ref{CE_alternative_constraint1} is always unbounded.
We can then conclude that there is at least one
correlated equilibrium in every trading period of the prediction market.
\end{proof}

\omitit{We have adapted the CE calculation algorithm \cite{Papa08}
to calculate the correlated equilibria in our prediction market. This
algorithm has a worst case running time polynomial in the number of strategies.}

The CE calculation algorithm used by the trading agents in our
agent-based prediction market is shown in Algorithm \ref{CEalg}.
The calculation of the matrix values of the $U$ matrix must be done once for each
of the $N$ agents. Using the ellipsoid algorithm \cite{Papa08}
the computation of the utility difference $u_i(\phi) - u_i(\phi_i,\phi_{-i})$
for each agent $i$ can be done in $|\Phi_i|^2$ time. Therefore, the time
complexity of the CECalc algorithm during each trading period comes to
$N \times |\Phi_i|^2$.

\begin{algorithm}[]
\textbf{CECalc}($D,\Phi$) \\
\KwIn{$D, \Phi$ //$D$ is the duration of the market, $\Phi$ is the set of strategies }
\KwOut{$p$ //correlated equilibrium }
\ForEach{$t \leftarrow 0$ to $D$ }{//do this in each trading period\\
Let $U$ be the matrix  consisting of the values of $(u_i(\phi) - u_i(\phi_i, \phi_{-i}))$, $\forall i \in N, \phi \in \Phi, \phi_{-i} \in \Phi_i$ \\
    $p_t' \leftarrow$ getDualDistribution($\Phi,U$);\\
    $p_t \leftarrow$ solve for $p_t$ s.t. $p_t U^T p_t' = 0$;\\
	return $p_t$;\\
}
$ $\\
\textbf{GetDualDistribution}($\Phi,U$)\\
\KwIn{$\Phi,U$}
\KwOut{$\Delta$}
$l = 0;$\\
$p'_l \in [0,1]$;\\
$\Delta = \{ \}$;\\
\While{$U^T \cdot p_l' \leq -1$ is feasible}{
$\Delta = \Delta +p_l';$\\
$p_{l+1}' = p_l + \epsilon^N$;  //increase all elements of $p'$ by some small amounts from vector $\epsilon^N$\\
$l++;$\\
}
return $\Delta$;
\caption{Correlated Equilibrium Algorithm}
\label{CEalg}
\end{algorithm}

\subsection{Correlated Equilibrium with Trading Agents'
Risk Preferences}
\label{sec_riskpref}
Incorporating the risk preferences of the trading agents
is an important factor in prediction markets.
For example, the erroneous result related to the non-correlation
between the trader beliefs and market prices in a prediction
market in \cite{Manski06} was because the risk preferences of
the traders were not accounted for, as noted in \cite{Gjerstad06}.
This problem is particularly relevant for risk averse
traders because the beliefs(prices) and risk preferences
of traders have been reported to be directly correlated \cite{Dimitrov10,Kadane}.
In this section, we examine CE that provides truthful
revelation incentives in the prediction market where the agent
population has different risk preferences.
The risk preference of an agent $i$ is modeled through a utility
function called the constant relative risk aversion (CRRA).
We use CRRA utility function to model risk averse agents because
it allows to model the effect of different levels of risk aversion and it
has been shown to be a better model than alternative families of risk modeling utilities \cite{Wakker}.
It has been widely used for modeling
risk aversion in various domains including economic domain \cite{Holt}, psychology \cite{Luce} and in the health domain \cite{Bleichrodt}.
The CRRA utility function is given below:
\begin{eqnarray}
 u_i(R_i) & = & \frac{R_i^{1-\theta_i}}{1-\theta_i}, \,\, \mbox{if}\, \theta_i \neq 1
 \nonumber \\
& = & ln (R_i), \,\, \mbox{if}\, \theta_i = 1
 \nonumber
\end{eqnarray}
Here, $-1<\theta_i < 1$ is called the risk preference factor of agent
$i$ and $R_i$ is the payoff or reward to agent $i$ calculated
using the LMSR \cite{Chen07}. When $R_i < 0$, we
may get a utility in the form of a complex number. In that case,
we convert the complex number to a real number by calling an existing
function that uses magnitude and the angle of the complex number for conversion.
For a risk
neutral agent, $\theta_i=0$, which makes $u_i = R_i$, as we have
assumed for our CE analysis of the risk neutral agent in
Section \ref{sec_CE}. For risk averse agents, $\theta_i >0$.
This makes the CRRA utility function concave. But the
concave structure of the utility $u_{i}$ for a risk
averse agent does not affect the \textit{existence} of at least
one correlated equilibrium because the unboundedness of
Equation \ref{CE_alternative_constraint1}
is not affected by the concave structure of $u_i$.
However, the equilibrium obtained with the risk averse
utility function can be different from the one obtained
with the risk neutral utility function - because the
best response $p(\phi)$ to $p(\phi_{-i})$
calculated with the risk neutral utility function might
not remain the best response when the agents are risk averse
and use the risk averse utility function.
To find a correlated equilibrium $p(\phi)$
in the market with risk averse agents, we first characterize
the set of all Pareto optimal strategy profiles.
A strategy profile $\phi^P$ is Pareto optimal if there does not exist another
strategy profile $\phi'$ such that $u_i(\phi') \geq u_i(\phi^P)$
$\forall i \in N$ with at least one inequality strict.
In other words, a Pareto optimal strategy profile is one such that
no trader could be made better off without making someone else worse off.
A Pareto optimal strategy profile can be found by maximizing weighted utilities
\begin{equation}
max_{\phi} \sum_{i=1}^{|N|} \lambda_i u_i(\phi) \mbox{ for some } \lambda_i
\label{max}
\end{equation}
Setting $\lambda_i=1$ for all $i \in N$ gives
a utilitarian social welfare function. The maximization problem in Equation \ref{max} can be solved
using the Lagrangian method. We get the following system of $|N|$ equations:
\begin{equation}
\sum_{i=1}^{|N|} \lambda_i \frac{u_i(\phi)}{\phi_i} = 0 \mbox{ , } \forall j=1,...,|N|
\label{FOC}
\end{equation} that must hold at $\phi^P$.
Each of these equations is obtained by taking a partial
derivative of the respective agent's weighted utility with respect to respective agent's strategy profile,
thus solving the maximization problem given in Equation \ref{max}.
By solving the system of equations \ref{FOC} we get the set of Pareto optimal
strategy profiles.
To determine if $p(\phi)$ is a correlated equilibrium we simply check whether it satisfies
correlated equilibrium constraints for Pareto optimal strategy profile given below:\\
$\sum_{\phi_{-i} \in \Phi_{-i}} p(\phi) (u_{i}(\phi) - u_{i}(\phi_i', \phi_{-i})) \geq 0,  \forall i \in N, \forall \phi_i, \phi_i' \in \Phi_i$

\begin{proposition}
If $p$ is a correlated equilibrium and $\phi^P$ is a Pareto optimal
strategy profile calculated by $p$ in a prediction market with
risk averse agents,
then the strategy profile $\phi^P$ is incentive compatible, that is
each agent is best off reporting truthfully.
\end{proposition}

\begin{proof}
We prove by contradiction. Suppose that $\phi^P$ is not an incentive compatible strategy, that is, there is some other $\phi'$ for which
\begin{equation}
u_i(\phi') \geq u_i(\phi^P)
\label{contradiction}
\end{equation}
Equation \ref{contradiction} violates two properties of $\phi^P$. First of all, since $\phi^P$ is Pareto optimal,
 we know that Equation \ref{contradiction} is not true, since $u_i(\phi^P) \geq u_i(\phi')$ by the definition
 of Pareto optimal strategy profile. Secondly, if we rewrite Equation \ref{contradiction} as
 $u_i(\phi^P) - u_i(\phi') \leq 0$ and multiply both sides by $p(\phi^P)$, we get
 $p(\phi^P) [u_i(\phi^P) - u_i(\phi')] \leq 0$. Since
 $p$ is a correlated equilibrium this inequality can not hold, otherwise it would violate
 the definition of the correlated equilibrium.
\end{proof}

\section{Experimental Results}
We have conducted several simulations using our POSGI prediction market.
The main objective
of our simulations is to test whether there is a benefit to the agents to follow
the correlated equilibrium strategy.
We do this by analyzing
the utilities of the agents and the market price.
The default values for the statistical distributions
for market related parameters such as the number
of days over which the market runs representing the duration
of an event and the information signal
arrival rate  were taken from data obtained from the Iowa Electronic
Marketplace(IEM) movie market for the event \textit{Monsters Inc.}
\cite{IEMURL} and are shown in Table \ref{table_sim_parms}.
\begin{table}[t]
\begin{center}
\begin{tabular}{|l|l|}
\hline
\textbf{Name} & \textbf{Value}\\
\hline
No. of days  & $50$\\
\hline
Duration of $1$ day & $1 sec$\\
\hline
No. of agents & $2$\\
\hline
No. of events & $1$\\
\hline
Max. allowed buy/sell quantity of security & $1$\\
\hline
Price of security & $[0,1]$\\
\hline
\end{tabular}
\caption{Parameters used for our simulation experiments}
\label{table_sim_parms}
\end{center}
\end{table}

For all of our simulations, we assume there is one event in the market with two
outcomes (positive or negative); each outcome has one security associated with it.
We consider events that are disjoint (non-combinatorial). This allows us to compare our proposed strategy empirically with other existing
strategies while using real data collected from  the IEM, which also considers
non-combinatorial events. Since we consider disjoint events,
having one event vs. multiple events does not change the strategic behavior of the agent or the results
of each agent{\footnote{We have verified positively that our system can scale effectively with the number of events and agents.}}.
We report the market price for the security
corresponding to the outcome of the event being $1$ (event occurs).

We compare the trading agents' and market's behavior under various strategies employed by the agents.
We use the following five  well-known strategies for comparison \cite{Ma08}
which are distributed uniformly over the agents in the experiments.
\begin{enumerate}
\item ZI (Zero Intelligence) - each agent submits randomly
calculated quantity to buy or sell.
\item ZIP (Zero Intelligence Plus) - each agent selects a quantity to buy or sell that
satisfies a particular level of profit by adopting its profit margin
based on past prices.
\item CP (by Preist and Tol) - each agent adjusts its quantity to buy or sell
based on past prices and tries to choose that quantity so that it is competitive among other agents.
\item GD (by Gjerstad and Dickhaut) - each agent maintains a
history of past transactions and chooses the quantity to buy or sell that
maximizes its expected utility.
\item DP (Dynamic Programming solution for POSG game) - each
agent uses dynamic programming solution to find the best
quantity to buy or sell that maximizes its expected utility given past prices,
past utility, past belief and the information signal
\cite{Hansen04}.
\item CE (Correlated equilibrium solution) - each agent uses
Algorithm \ref{CEalg} from Section \ref{sec_CE} to determine a correlated
equilibrium strategy that gives a recommended quantity to buy or sell.
\end{enumerate}

\begin{figure}[h]
\begin{center}
\begin{tabular}{lll}
\hspace{-0.2in}
\includegraphics[width=2.0in]{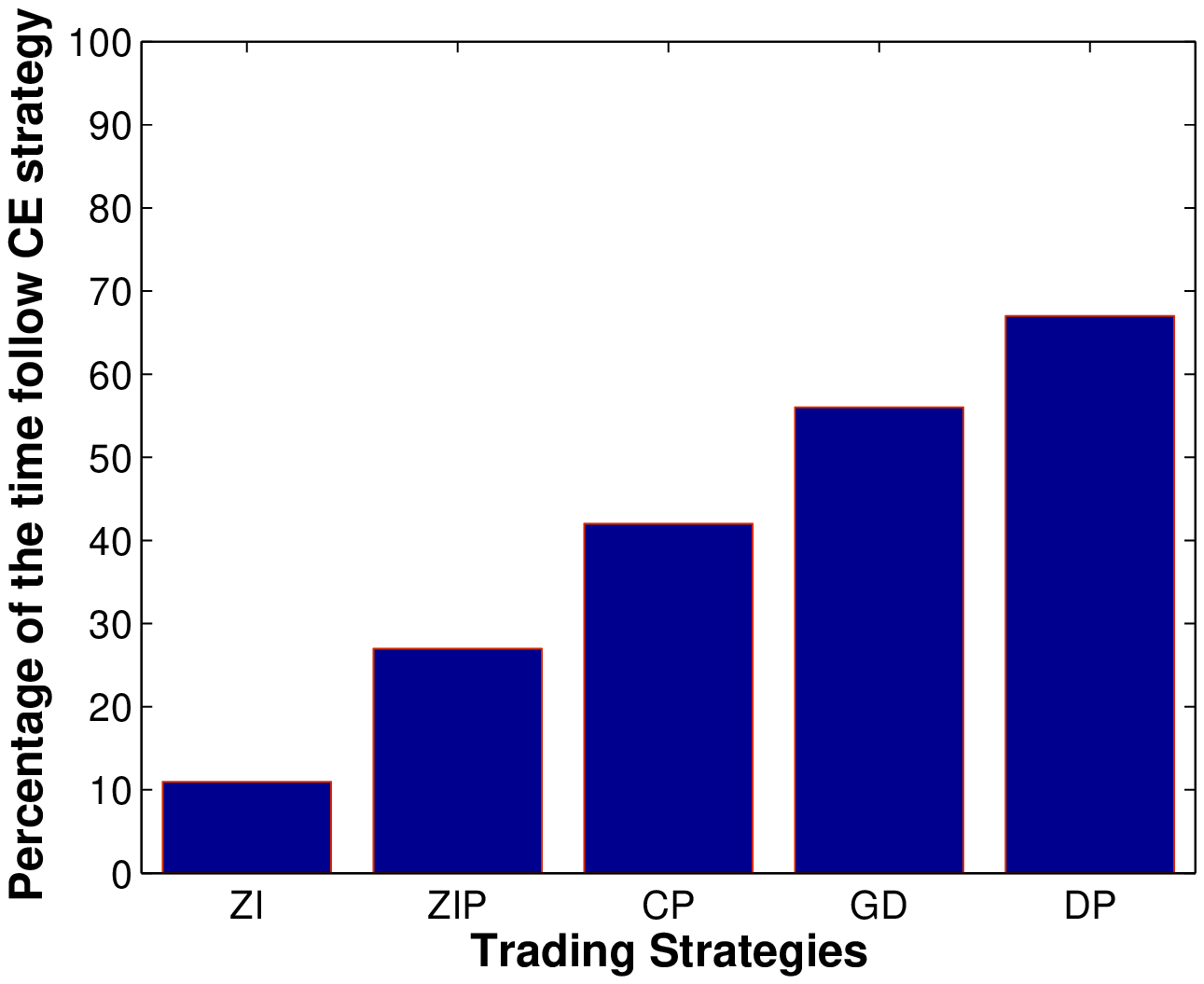}
\hspace{-0.2in}
&
\hspace{-0.2in}
\includegraphics[width=2.0in]{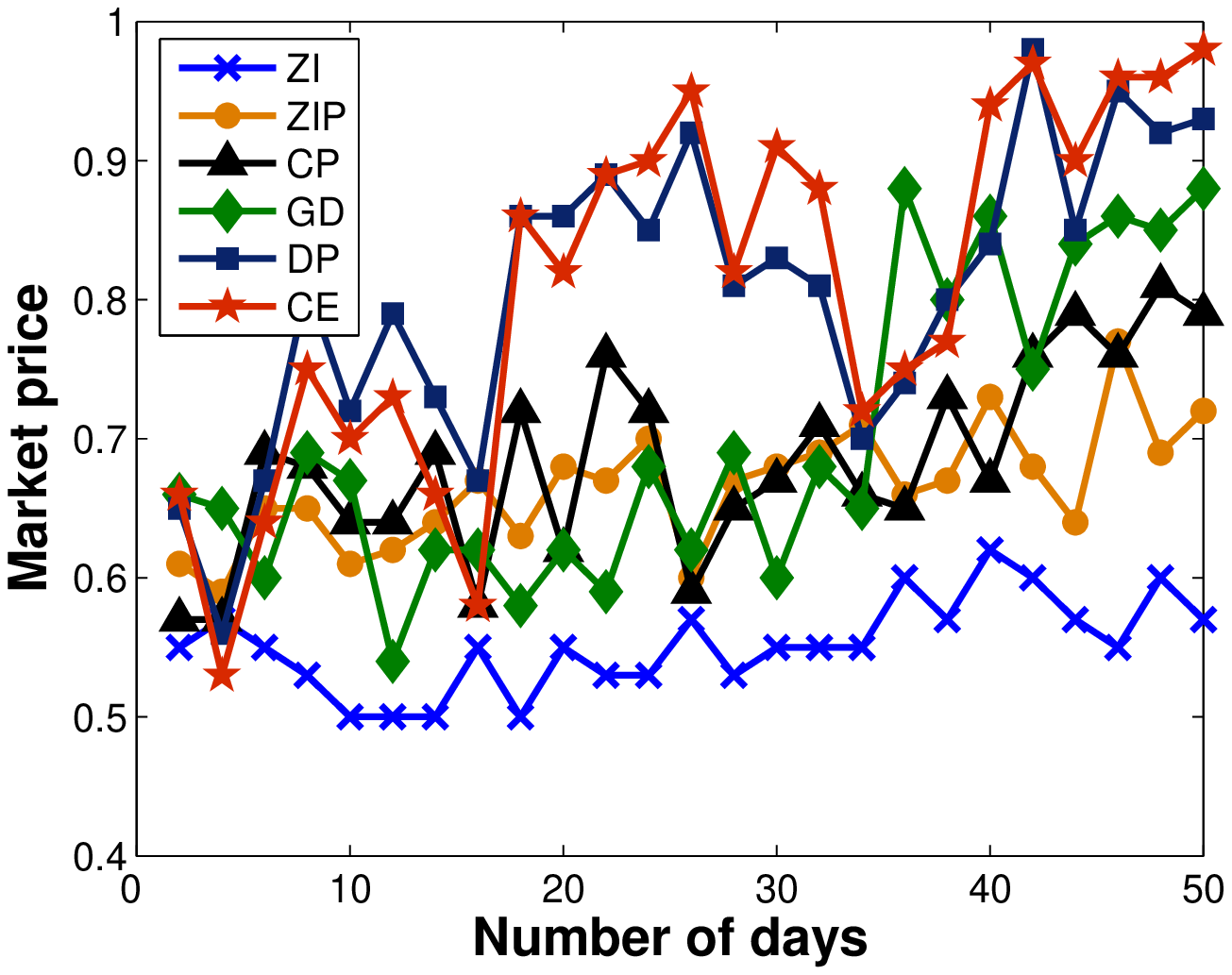}
\hspace{-0.2in}
&
\hspace{-0.2in}
\includegraphics[width=2.0in]{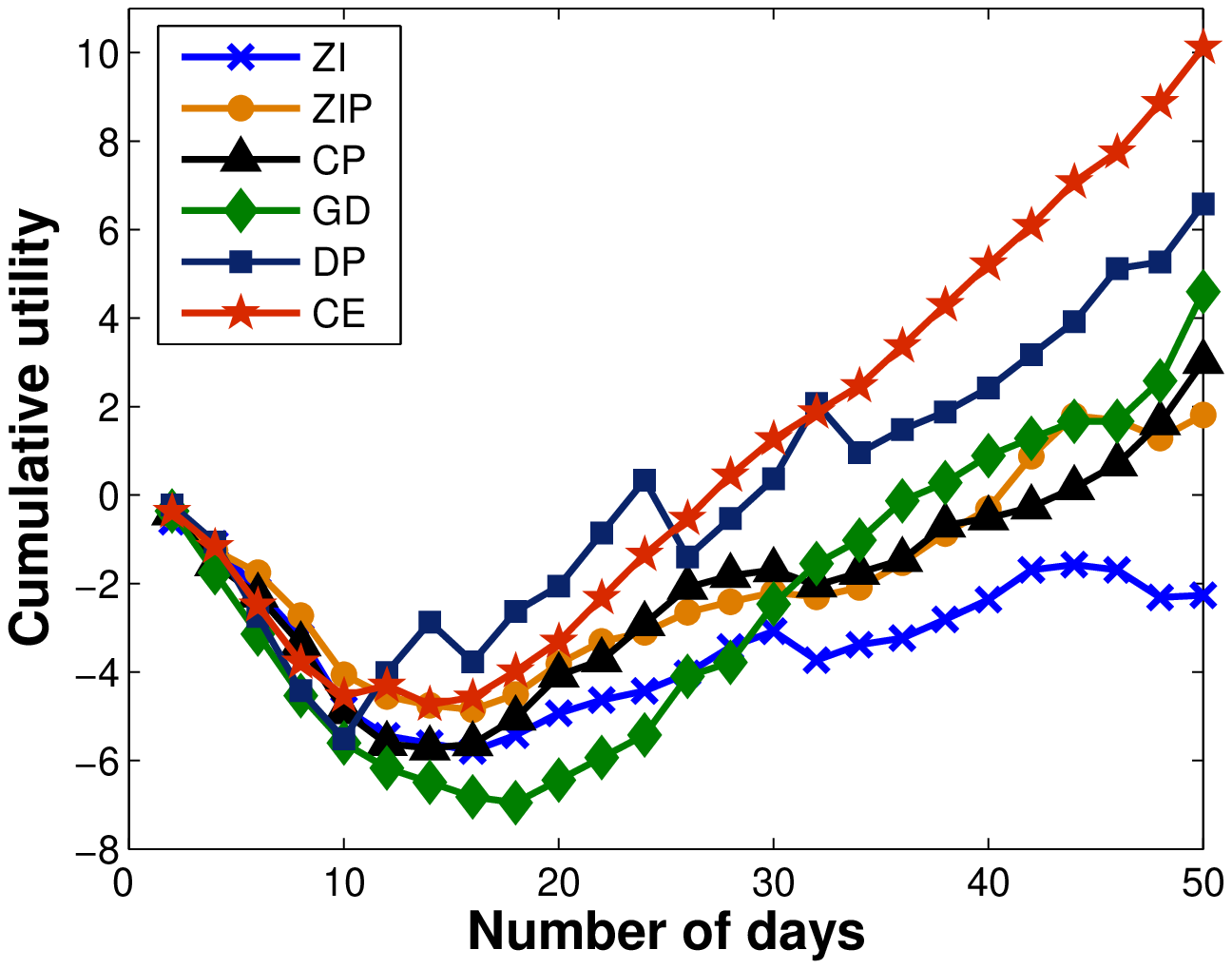}
\hspace{-0.2in}\\
\hspace{0.8in}a & \hspace{0.8in}b & \hspace{0.8in}c
\end{tabular}
\caption{(a) Percentage of the number of days each compared trading strategy follows a correlated equilibrium strategy, (b) Market Prices, and, (c) Utilities of the risk neutral agent while using different trading strategies.}
\label{price}
\end{center}
\end{figure}

\subsection{Risk-neutral Agents}
For our first group of experiments we assume that the agents are risk neutral.
First, we attempt to understand the
correspondence between our CE strategy and each of the other
compared strategies.
For comparing the CE strategy with each strategy, we ran the simulations with
identical settings, once with agents using the CE strategy for
making trading decisions and then with the same agents using the
compared strategy for making trading decisions.
Figure \ref{price}(a) shows the number of times(days) an action
recommended to a trading agent by the CE strategy is the same
as the action recommended to a trading agent by one of the compared
strategies, expressed as a percentage.
We observe that the action recommended to a trading agent
employing the ZI strategy is the same as the action recommended to a trading
agent using CE strategy only $11\%$ of the time, while for the trading agent using
 the more refined DP  strategy it is the same $68\%$ of the time. The higher percentage
of adoption of the CE strategy by the more refined strategies indicates
that the DP strategy is more sophisticated than simpler strategies.

In our next set of experiments, we compare the market prices and utilities obtained by a trading agent over time (days) while using different trading strategies. In each simulation run, the trading agent uses one of the compared strategies for determining its action at each time step (day). We compare the cumulative market prices and utilities for each such simulation.
We report a scenario where the event in the prediction market happens and show the
market prices and utilities for the security corresponding to the positive outcome of the event.
In this scenario, the market price should approach $\$1$ (event happens) as the prediction market's duration nears end.
Figure \ref{price}(b) shows the market prices of the orders placed
by the trading agents during the duration of the event for different
strategies.  We observe that agents using the CE strategy perform better
since they are able to
trade at prices that are closer to $\$1$,
indicating that agents using the CE strategy are able to respond
to other agents' strategies and predict the aggregated price of the
security more efficiently. This efficiency is further supported by the graph in
Figure \ref{price}(c) that shows the utility of the agents
while using different strategies. We see that the agents using
the CE strategy are finally able to obtain $38\%$ more utility
than the agents following the next best performing strategy (DP),
and $85\%$ more utility than the agents following the second worst
performing strategy (ZIP).

\begin{figure}[h]
\begin{center}
\begin{tabular}{ll}
\hspace{-0.15in}
\includegraphics[width=2.5in]{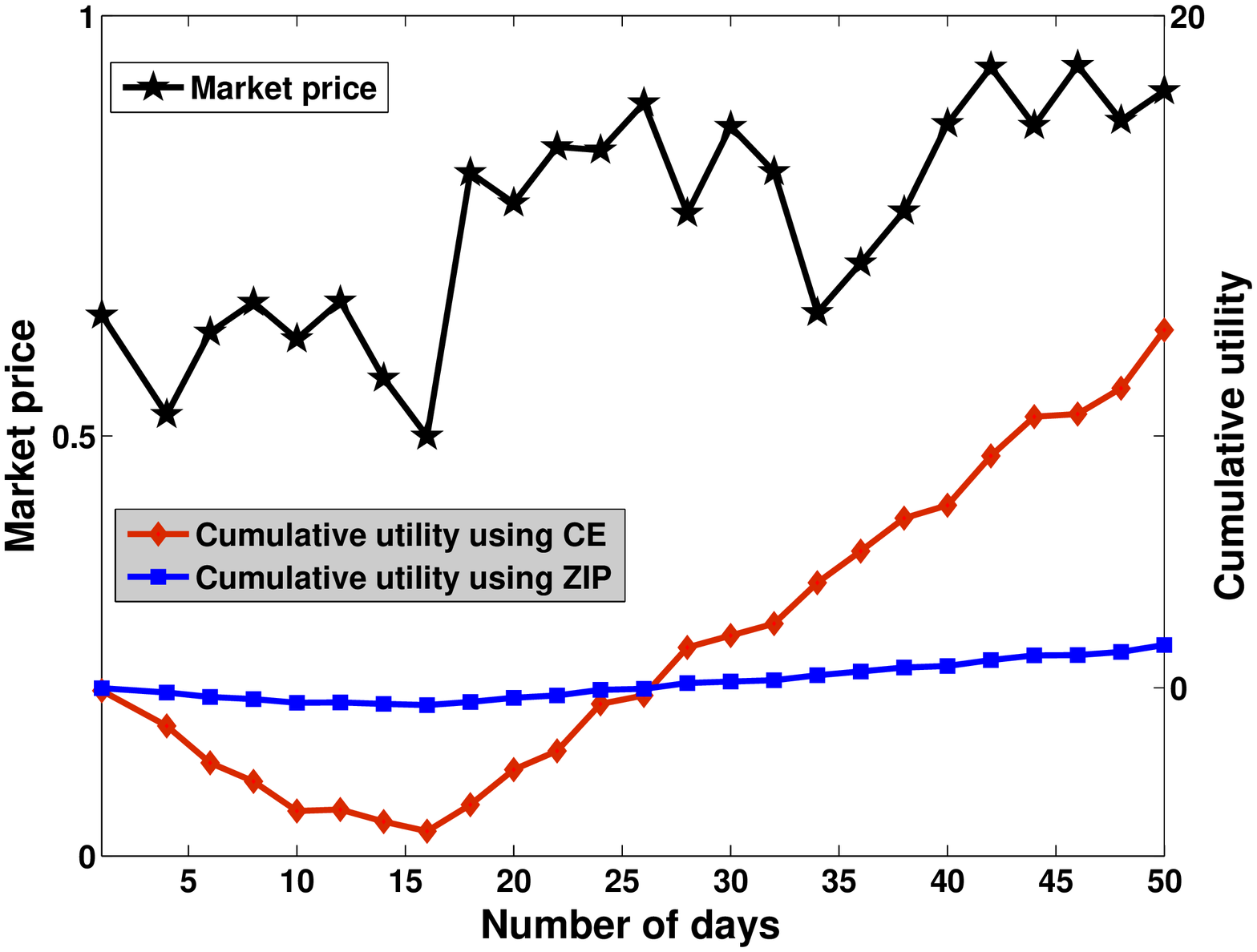}
\hspace{-0.15in}
&
\hspace{-0.15in}
\includegraphics[width=2.7in]{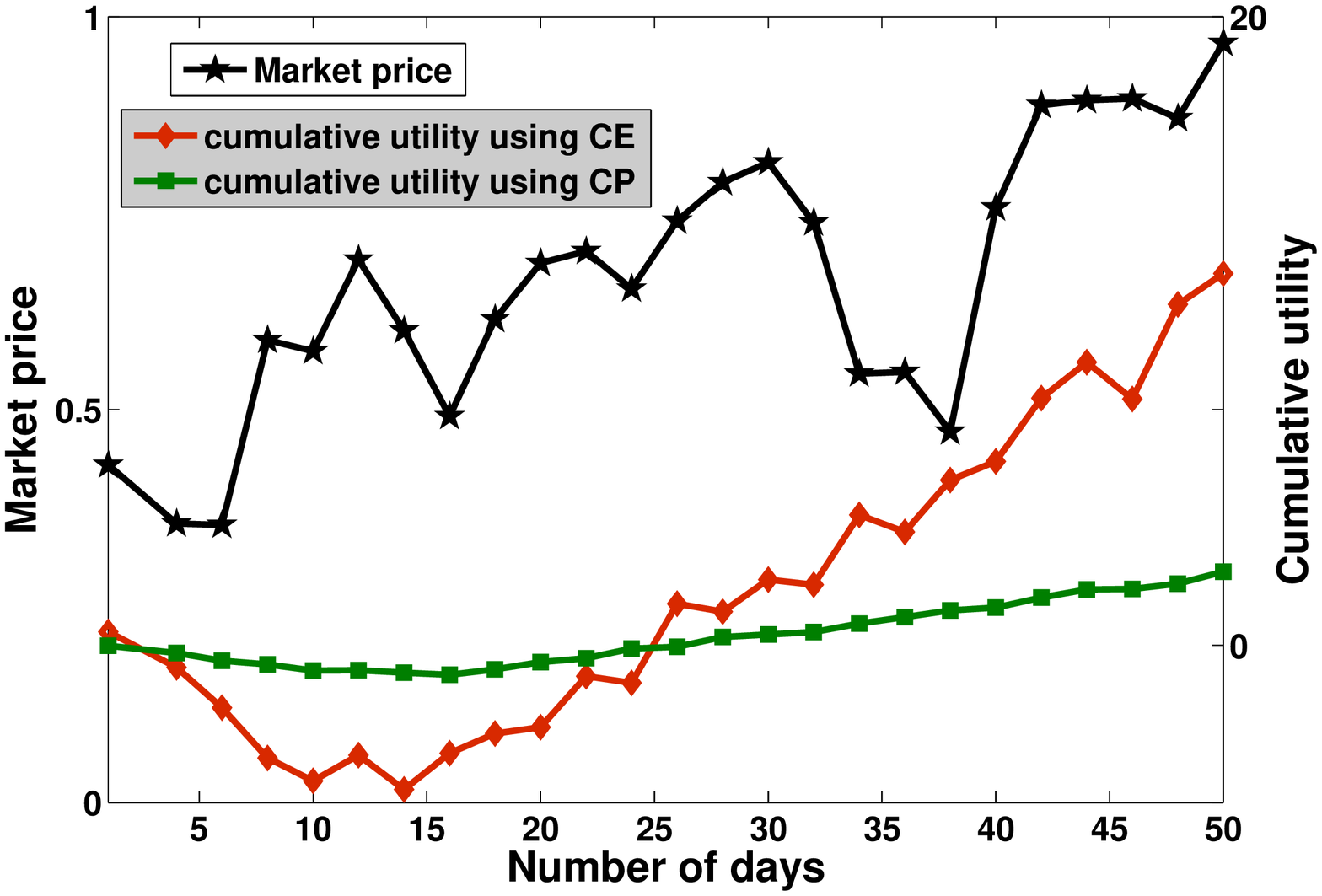}
\hspace{-0.15in}\\
\hspace{0.8in}a & \hspace{0.8in}b 
\end{tabular}
\caption{Comparison of the cumulative market prices and utilities obtained by: (a) CE and ZIP strategies, and, (b) CE and CP strategies.}
\label{CEvs1}
\end{center}
\end{figure}

\begin{figure}[h]
\begin{center}
\begin{tabular}{ll}
\hspace{-0.15in}
\includegraphics[width=2.5in]{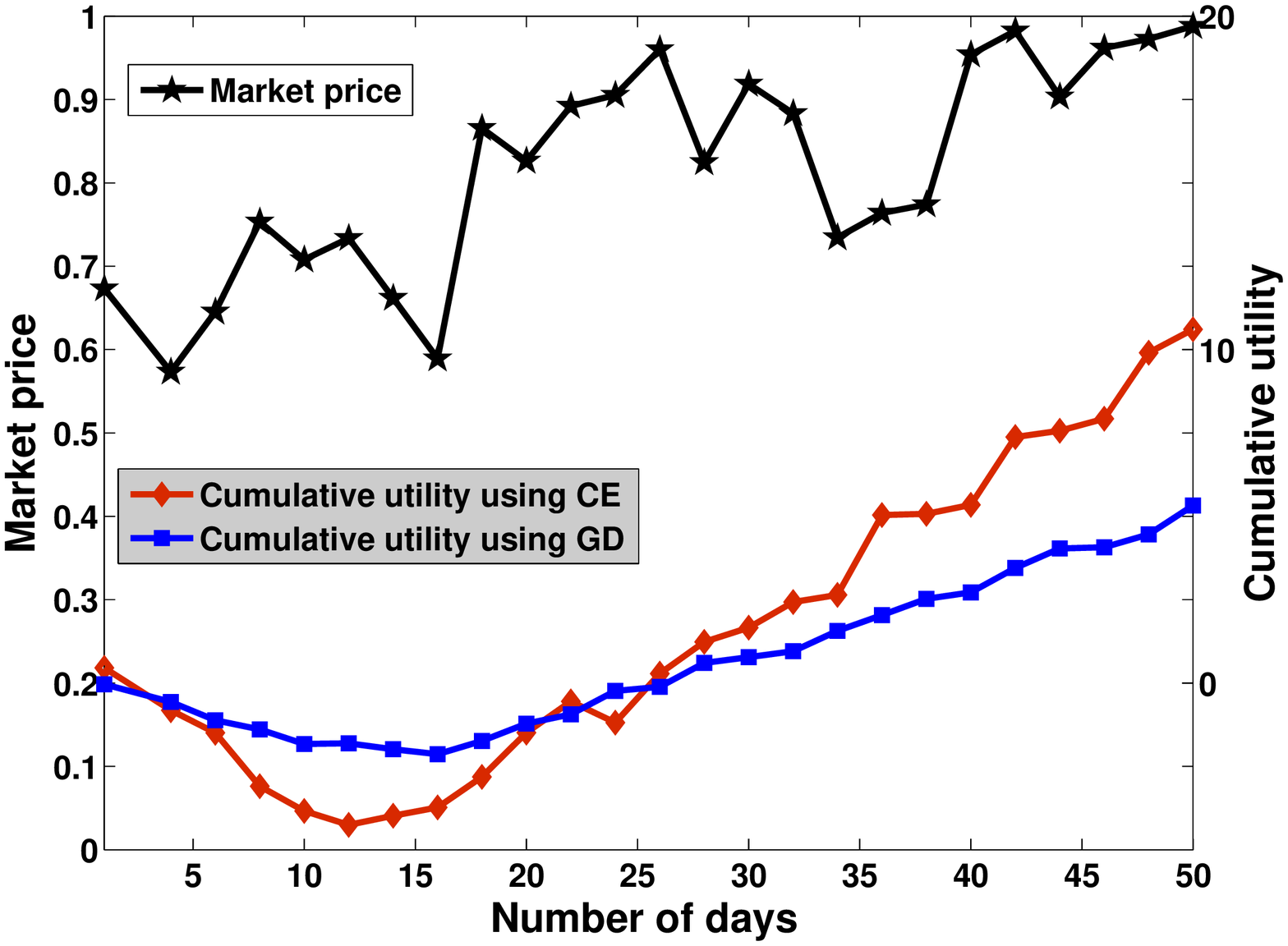}
\hspace{-0.15in}
&
\hspace{-0.15in}
\includegraphics[width=2.7in]{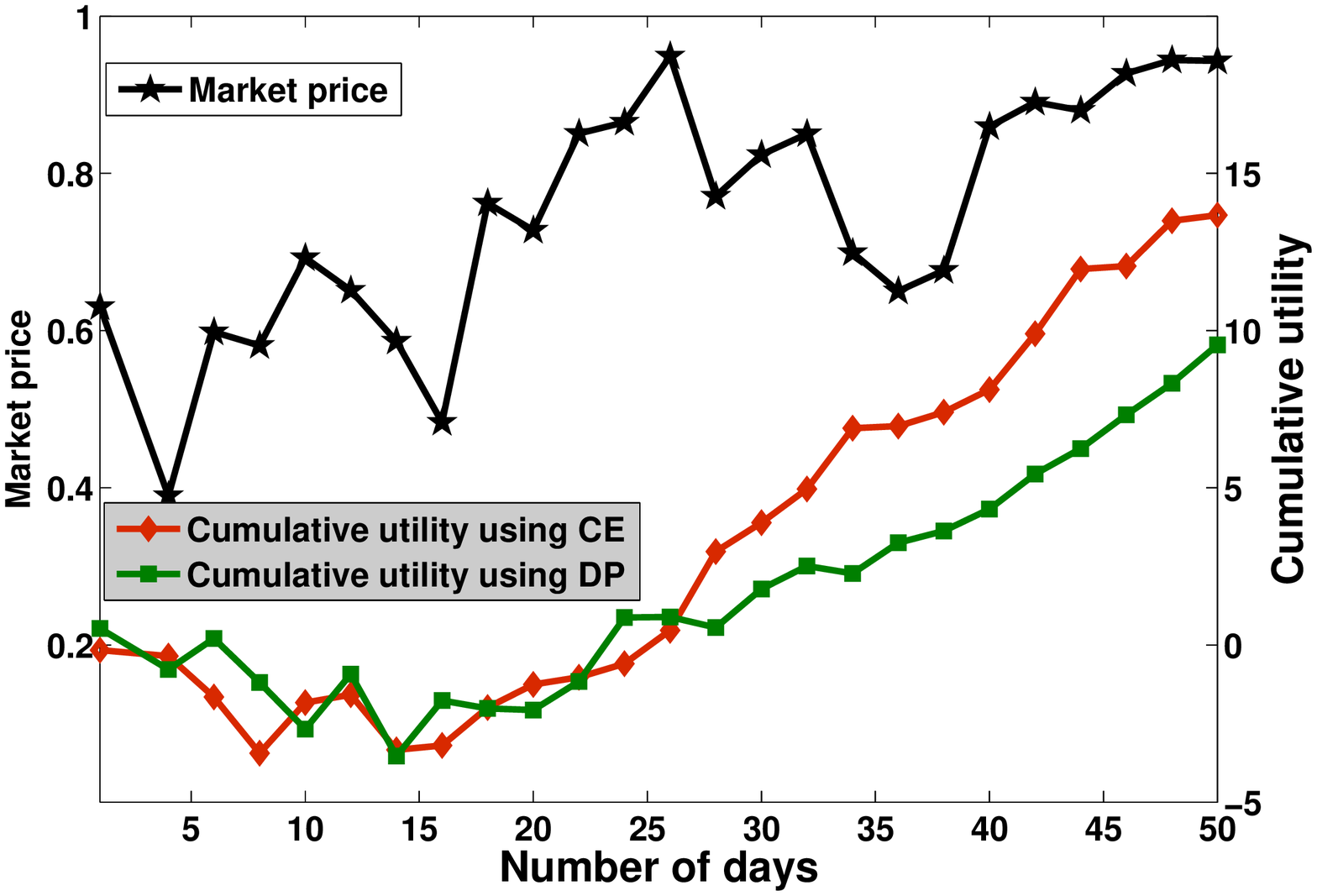}
\hspace{-0.15in}\\
\hspace{0.8in}a & \hspace{0.8in}b 
\end{tabular}
\caption{Comparison of the cumulative market prices and utilities obtained by: (a) CE and GD strategies, and, (b) CE and DP strategies.}
\label{CEvs2}
\end{center}
\end{figure}

For our following set of experiments, we further compare the CE strategy with
other strategies by using two agents in the same simulation run - one agent using CE strategy
and the other agent using the compared strategy.
The results are reported in Figure \ref{CEvs1} and \ref{CEvs2}.
We can see that in each scenario the CE strategy outperforms other strategies. From Figures \ref{CEvs1}(a) and \ref{CEvs2}(b) respectively, we observe that the trading agent using CE strategy gets $95\%$ more utility than the trading agent using ZIP strategy and $41\%$ more utility than the trading agent using DP strategy.
Overall, we can say that using the POSGI
model allows the agent to avoid myopically predicting prices and use
the correlated equilibrium to calculate prices more accurately
and obtain higher utilities.

\subsection{Risk averse Agents}
\begin{figure}[h]
\begin{center}
\begin{tabular}{lll}
\hspace{-0.15in}
\includegraphics[width=2.0in]{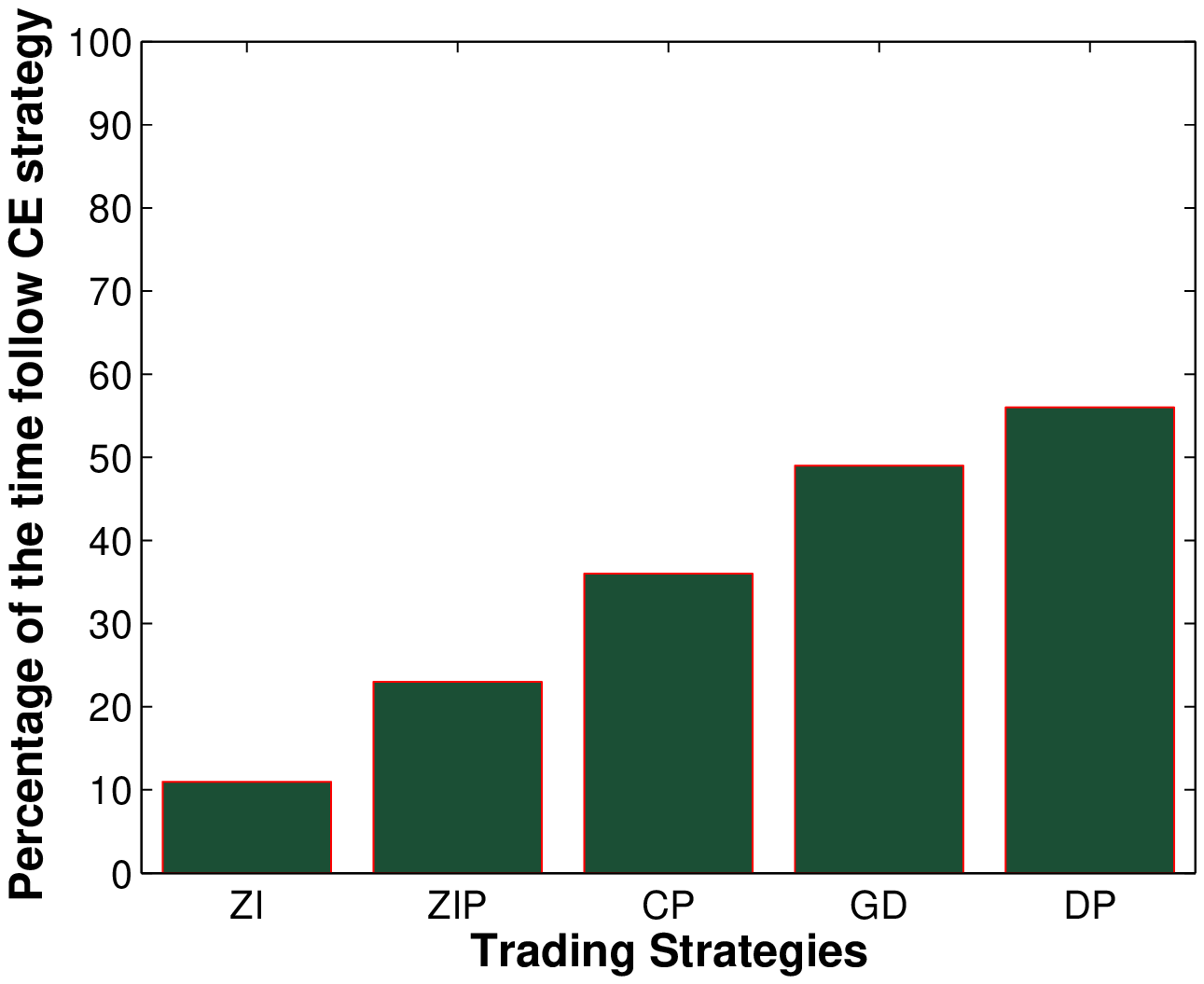}
\hspace{-0.2in}
&
\hspace{-0.2in}
\includegraphics[width=2.0in]{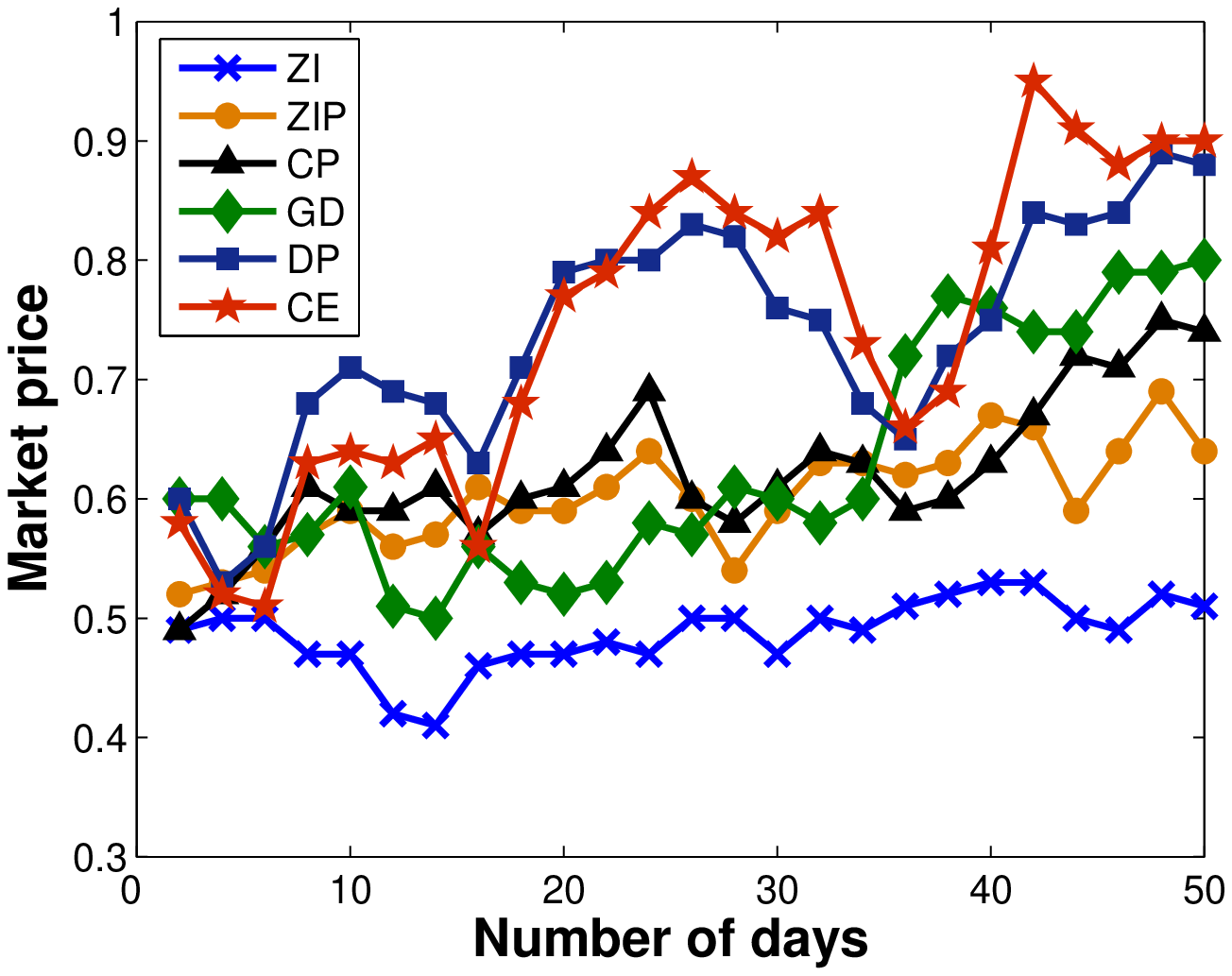}
\hspace{-0.2in}
&
\hspace{-0.2in}
\includegraphics[width=2.0in]{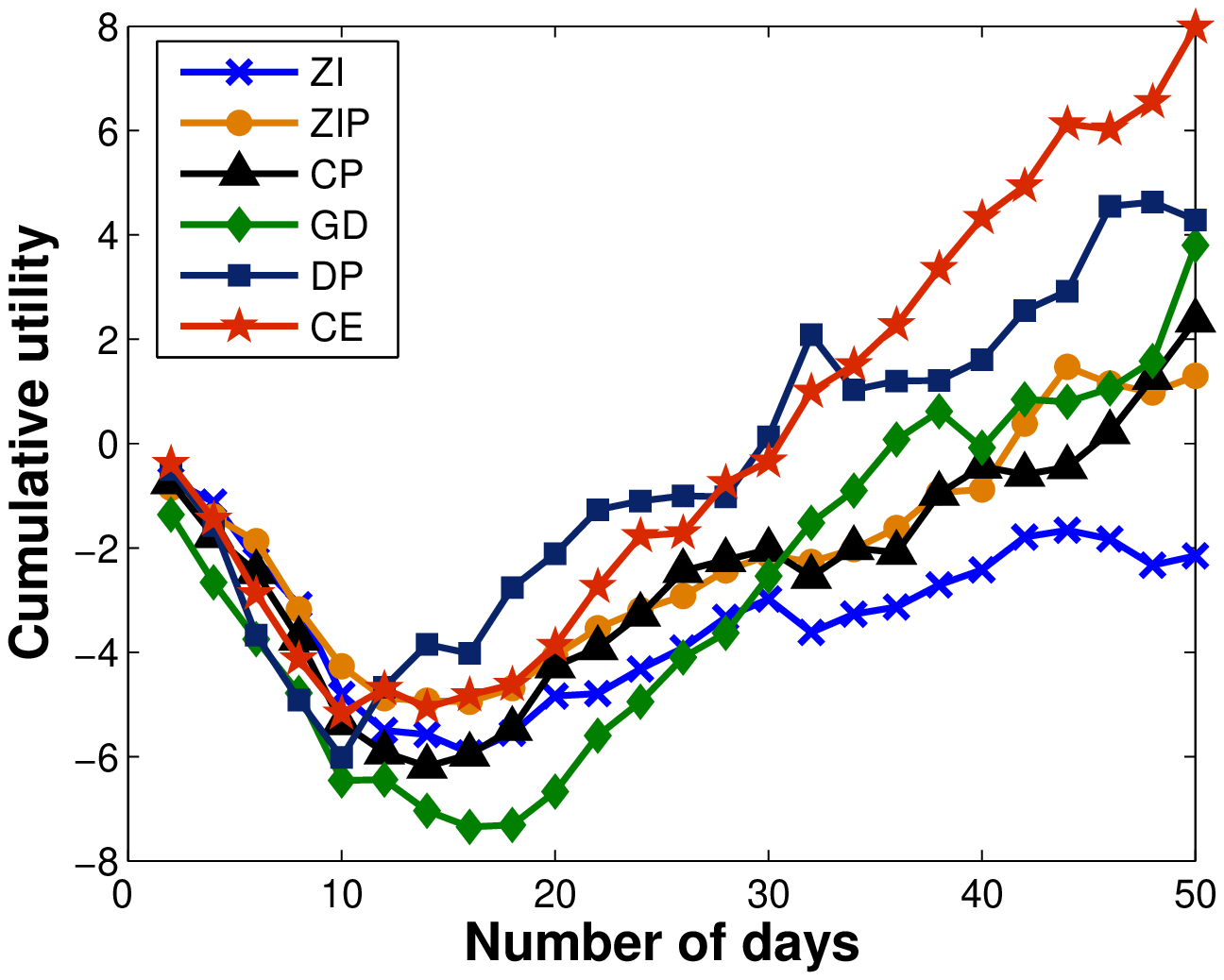}
\hspace{-0.2in}\\
\hspace{0.8in}a & \hspace{0.8in}b & \hspace{0.8in}c
\end{tabular}
\caption{(a) Percentage of the number of days each compared trading strategy follows a correlated equilibrium strategy, (b) Market Prices, and, (c) Utilities of the risk averse agent while using different trading strategies.}
\label{price_ra}
\end{center}
\end{figure}
Next, we conduct similar simulations for the market while
using risk averse agents that use the risk averse utility
function given in Section \ref{sec_riskpref}.
Since in real world most people are risk averse, we ran several
experiments with different values of
the risk preference $\theta_i>0$. We report the results for $\theta_i=0.8$ (strongly risk averse trading agent) to see a better effect of the
risk averse preference of trading agents on the market prices and their own utilities.
To compare the CE strategy with each of the other strategies we ran the simulation with
identical settings, once with risk averse agents using the CE strategy for
making trading decision and then with the same risk averse agents using the
compared strategy for making trading decisions.
Figure \ref{price_ra}(a) shows the number of times an action
recommended to a risk averse trading agent by the CE strategy is the same
as the action recommended to a risk averse trading agent by one of the compared
strategies, expressed as a percentage.
We can see that the risk averse agents employing any of the five
compared strategies end up adopting the action that it the same as the action proposed by the CE strategy
only $11-60\%$ of the time. We notice that
the ZI agents are not affected by the risk averse behavior
of the agents (they follow the CE strategy $11\%$ of the time when
the agents are either risk averse or risk neutral)
because they do not use their past utility to determine
future prices. However, other compared agent strategies
use their past utilities to predict future prices.
The concave nature of the CRRA utility function with $\theta_i>0$
lowers the utilities that these agents receive, and,
this results in these agents following the CE strategy less often.
The effect of the lowered utility due to the risk averse (concave)
utility function is also seen in Figures \ref{price_ra}(b) and (c).
Figures \ref{price_ra}(b) shows the market prices
for an event that has a final outcome $=1$.
We again observe that the agents using CE strategy
are able to predict the aggregated price of the
security more efficiently during the duration of the event.
Figure \ref{price_ra}(c) that shows the utility of the risk averse agents
while using different strategies. We see that the agents using
the CE strategy are able to obtain $42\%$ more utility
than the agents following the next best performing strategy (DP)
and $127\%$ more utility than the agents following the worst
performing strategy (ZI).

\begin{table}[t]
\begin{center}
\begin{tabular}{|l|l|l|l|l|l|l|}
\hline
Strat. &\multicolumn{3}{l|}{Risk Neutral Agent}&\multicolumn{3}{l|}{Risk Averse  Agent}\\
\cline{2-7}
&$\%F_{CE}$& Price & Utility & $\%F_{CE}$& Price & Utility\\
\hline
ZI  &  $11\%$ & $36\%$ & $123\%$ & $11\%$ & $32\%$ & $132\%$\\
\hline
ZIP & $26\%$ & $25\%$ & $85\%$ & $22\%$ & $26\%$ & $89\%$\\
\hline
CP & $41\%$ & $23\%$ & $77\%$ & $35\%$ & $20\%$ & $83\%$\\
\hline
GD & $47\%$ & $19\%$ & $52\%$ & $39\%$ & $16\%$ & $58\%$\\
\hline
DP & $68\%$ & $9\%$ & $38\%$ & $60\%$ & $10\%$ & $41\%$\\
\hline
\end{tabular}
\caption{Percentage of the number of days each compared strategy follows the CE strategy ($\%FCE$ in cols. $2$ and $5$), the percentage of the difference between each compared strategy and the CE strategy for the market price (cols. $3$ and $6$) and utility (cols. $4$ and $7$). Results are shown for both risk neutral and risk averse trading agents. }
\label{results}
\end{center}
\end{table}

The summary of all of our results is given in Table \ref{results}.
For each type of agent, risk neutral or risk averse, the first column, $\%F_{CE}$, indicates percentage of the number of days the actions of the trading agents
using each compared strategy are the same as the CE strategy. The second and the third column show the
percentage of the difference between each compared strategy and the CE strategy for the market price and utility, respectively.
In summary, the POSGI model and the CE strategy result in
better price tracking and higher utilities because they
provide each agent with a strategic behavior while taking
into account the observations of the prediction market and
the new information of the events.

\section{Conclusions}
In this paper, we have described an agent-based POSGI
prediction market
with an LMSR market maker
and empirically compared different agent behavior
strategies in the prediction market. We proved the existence of a correlated equilibrium
in our POSGI prediction market with risk neutral agents,
and showed how correlated equilibrium can be obtained
in the prediction market with risk averse agents.
We have also empirically verified that when the
agents follow a correlated equilibrium strategy
suggested by the market maker they obtain
higher utilities and the market prices are more
accurate.

In the future we are interested in conducting experiments in an
$n$-player scenario for the POSGI formulation
given in Section $3$ using richer commercially available prediction market
data sets. We also plan to analyze the performance of the CE strategy in the scenarios where the outcome of the event does not correspond to the predicted outcome of the event by the market price, i.e. when prediction market fails.
We also plan to investigate the dynamics evolving from
multiple prediction markets that interact with each other.
Finally, we are interested in exploring truthful revelation
mechanisms that can be used to limit untruthful bidding
in prediction markets.

\section*{Acknowledgments}
This research has been sponsored as
part of the COMRADES project funded by the Office of
Naval Research, grant number $N000140911174$.

\bibliographystyle{abbrv}

\end{document}